\documentclass[a4paper,UKenglish,cleveref, autoref, thm-restate]{lipics-v2021}
\pdfoutput=1 
\hideLIPIcs  

\usepackage{paper_macros}
\usepackage{cite}

\title{The Expressive Power of Uniform Population Protocols with Logarithmic Space} 

\author{Philipp Czerner}{Technical University of Munich, Germany \and \url{https://nicze.de/philipp}}{czerner@in.tum.de}{https://orcid.org/0000-0002-1786-9592}{}

\author{Vincent Fischer}{Technical University of Munich, Germany}{vincent.fischer@tum.de}{https://orcid.org/0009-0009-3071-0736}{}

\author{Roland Guttenberg}{Technical University of Munich, Germany}{guttenbe@in.tum.de}{https://orcid.org/0000-0001-6140-6707}{}

\authorrunning{P.\ Czerner, V.\ Fischer and R.\ Guttenberg} 

\Copyright{Philipp Czerner, Vincent Fischer and Roland Guttenberg} 



\category{} 

\nolinenumbers 

\usepackage{xcolor}
\usepackage{stmaryrd}
\usepackage{tikz}
\usetikzlibrary{automata, positioning, arrows, petri, backgrounds}

\newcommand{\N}{\mathbb{N}}

\newcommand{\Z}{\mathbb{Z}}

\newcommand{\NL}{\mathsf{NL}}
\def\O{\mathcal{O}}

\newcommand{\Abs}[1]{\mathopen|#1\mathclose|}

\definecolor{niceredbright}{HTML}{bd0310}
\definecolor{nicebluebright}{HTML}{197b9b}
\definecolor{nicered}{HTML}{7f0a13}
\definecolor{niceblue}{HTML}{104354}
\definecolor{nicegreen}{HTML}{217516}
\definecolor{nicepurple}{HTML}{884bab}
\definecolor{nicebg}{HTML}{f6f0e4}
\definecolor{niceredlight}{HTML}{c9888d}
\definecolor{nicebluelight}{HTML}{78a4b8}
\definecolor{nicegreenlight}{HTML}{76de68}
\definecolor{nicepurplelight}{HTML}{bc87db}

\newcommand{\TraNs}{tra}
\newcommand{\TraName}[1]{\tag*{⟨\textsf{#1}⟩}\label{\TraNs:#1}}
\newcommand{\TraRef}[1]{\text{\ref{\TraNs:#1}}}

\usepackage{mathtools}

\newcommand{\Multiset}[1]{\{#1\}}
\newcommand{\Alphabet}[0]{\Sigma}
\newcommand{\Transitions}[0]{\delta}
\newcommand{\Input}[0]{I}
\newcommand{\Output}[0]{O}
\newcommand{\Config}[0]{C}
\newcommand{\Prot}[0]{\mathcal{P}}
\newcommand{\FiniteProt}[0]{\mathcal{P}}
\newcommand{\Card}[1]{|#1|}

\newcommand{\SPACE}[0]{\mathsf{SPACE}}
\newcommand{\SNSPACE}[0]{\mathsf{SNSPACE}}
\newcommand{\UNSPACE}[0]{\mathsf{UNSPACE}}
\newcommand{\UENC}[0]{\mathsf{UENC}}
\newcommand{\UNL}[0]{\mathsf{UNL}}
\newcommand{\NSPACE}[0]{\mathsf{NSPACE}}
\newcommand{\coNSPACE}[0]{\mathsf{coNSPACE}}
\newcommand{\polylog}[0]{\operatorname{polylog}}
\newcommand{\StateComplexity}[0]{S}
\newcommand{\UPP}[0]{\mathsf{UPP}}
\newcommand{\NUPP}[0]{\mathsf{WUPP}}

\EventEditors{John Q. Open and Joan R. Access}
\EventNoEds{2}
\EventLongTitle{4th Symposium on Algorithmic Foundations of Dynamic Networks}
\EventShortTitle{SAND 2025}
\EventAcronym{SAND}
\EventYear{2025}
\EventDate{June 9--June 11, 2025}
\EventLocation{Liverpool, UK}
\EventLogo{}
\SeriesVolume{42}
\ArticleNo{23}

\begin{document}

\maketitle

\begin{abstract}
Population protocols are a model of computation in which indistinguishable mobile agents interact in pairs to decide a property of their initial configuration. Originally introduced by Angluin et.\ al.\ in 2004 with a constant number of states, research nowadays focuses on protocols where the space usage depends on the number of agents. The expressive power of population protocols has so far however only been determined for protocols using $o(\log n)$ states, which compute only semilinear predicates, and for $\Omega(n)$ states. This leaves a significant gap, particularly concerning protocols with $\Theta(\log n)$ or $\Theta(\operatorname{polylog} n)$ states, which are the most common constructions in the literature. In this paper we close the gap and prove that for any $\varepsilon>0$ and $f\in\Omega(\log n)\cap\mathcal{O}(n^{1-\varepsilon})$, both uniform and non-uniform population protocols with $\Theta(f(n))$ states can decide exactly those predicates, whose unary encoding lies in $\mathsf{NSPACE}(f(n) \log n)$.
\end{abstract}

\section{Introduction}

Population protocols are a model of computation in which indistinguishable mobile agents randomly interact in pairs to decide whether their initial configuration satisfies a given property. The decision is taken by \emph{stable consensus}; eventually all agents agree on whether the property holds or not, and never change their mind again. While originally introduced to model sensor networks \cite{ADFP04}, population protocols are also very close to chemical reaction networks \cite{SoloveichikCWB08}, a model in which agents are molecules and interactions are chemical reactions. 

Originally agents were assumed to have a finite number of states \cite{ADFP04, AngluinAE06, AngluinAER07}, however many predicates then provably require at least \(\Omega(n)\) time to decide \cite{DotyS15, BellevilleDS17, AlistarhAEGR17}, as opposed to recent breakthroughs of \(\O(\log n)\) time using \(\O(\log n)\) or even fewer states for important tasks like leader election \cite{BerenbrinkGK20} and majority  \cite{DotyEGSUS21}. Limitting the number of states to logarithmic is important in most applications, especially the chemical reaction setting, since a linear in \(n\) number of states would imply the unrealistic number of approximately \(10^{23}\) different chemical species. Therefore most recent literature focuses on the polylogarithmic time and space setting, and determines time-space tradeoffs for various important tasks like majority \cite{AlistarhGV15, AlistarhAEGR17, AlistarhG18, ElsasserR18, BerenbrinkEFKKR21, DotyEGSUS21}, leader election \cite{AlistarhAEGR17, ElsasserR18, BerenbrinkGK20} or estimating/counting the population size \cite{DotyEMST18, DotyE19, BerenbrinkKR19, DotyE21, DotyE22}.

This leads to the interesting open problem of characterizing the class of predicates which can be computed in polylogarithmic time using a logarithmic or polylogarithmic number of states. There is however a fundamental problem with working on this question: Despite the focus on \(\O(\log n)\) number of states in recent times, the expressive power for this number of states has not yet been determined. While it is known that protocols with \(o(\log n)\) number of states can only compute semilinear predicates  \cite{AngluinAER07, ChatzigiannakisMNPS11} and with \(f(n) \in \Omega(n)\) states the expressive power is \(\UNSPACE(n \log f(n))\) \cite{ChatzigiannakisMNPS11}, i.e.\ predicates which can be decided in \(\NSPACE(n \log f(n))\), when the input is encoded in unary, the important case of having logarithmically many states is unknown. To the best of our knowledge, the only research in this direction is \cite{BournezCR18}, where the expressive power is characterised for \(\polylog(n)\) number of states for a similar model --- not population protocols themselves. Their results do not lead to a complete characterization for \(\Theta(\log n)\) states since their construction is slightly too space-inefficient, simulating a \(\log \log n\)-space TM by approximately \(\log^2 n\) space protocols.

In this paper, we resolve this gap by proving that for functions \(f(n) \in \Omega(\log n) \cap \O(n^{1-\varepsilon})\), where $\varepsilon>0$, we have \(\UPP(f(n))=\UNSPACE(f(n) \cdot \log n)\), i.e.\ predicates computable by population protocols using \(\O(f(n))\) number of states are exactly the predicates computable by a non-deterministic Turing machine using \(\O(f(n) \cdot \log n)\) space with the input encoded in unary. The “\textsf{U}” in \(\UPP(f(n))\) stands for \emph{uniform}: Modern population protocol literature distinguishes between uniform and non-uniform protocols. In a non-uniform protocol, a different protocol is allowed to be used for every population size. While we have stated the expressive power for uniform protocols here, our complexity characterization also holds for non-uniform population protocols. 

Our results complete the picture of the expressive power of uniform protocols: For \(o(\log n)\) only semilinear predicates can be computed (open for non-uniform), for a class of \emph{reasonable} functions \(f(n) \in \Omega(\log n) \cap \O(n^{1-\varepsilon})\) , which contains most practically relevant functions\footnote{This will be clarified in the next section} we have \(\UNSPACE(\log (n) \cdot f(n))\) by our results, and for \(f \in \Omega(n)\) we have \(\UNSPACE(n \cdot \log f(n))\). (A slight gap between $\O(n^{1-\varepsilon})$ and $\Omega(n)$ remains.)

\medskip
\noindent\textbf{\textsf{Main Contribution.}}
The technically most involved part of our result is the lower bound, i.e.\ constructing a \(\O(f(n))\) space uniform population protocol simulating a \(\O(f(n) \log n)\) space Turing machine, or --- equivalently \cite{FischerMR68}, and used in our construction --- simulating a \(\O(2^{f(n) \log n})\)-bounded counter machine. Let us briefly illustrate the main techniques and difficulties towards this result. In a nutshell, the crucial difference between \(o(n)\) and $\Omega(n)$ states is the ability to assign unique identifiers to agents, and to store the population size $n$ in a single agent. In our construction, therefore, we must distribute the value of $n$ over multiple agents, and they must collaborate to compute operations involving it. We also introduce a novel approach for encoding the counters of the counter machine, as those described in previous publications such as \cite{AngluinAE06} and \cite{BournezCR18} cannot encode large enough numbers for our purposes.

\medskip
\noindent\textbf{\textsf{Overview.}}
The paper is structured as follows: In Section \ref{SectionPreliminaries} we give preliminaries and define population protocols. Section~\ref{sec:mainresult} briefly states our main result and prove the lower bound for weakly uniform poulation protocols. The proof of the matching upper bound (even for uniform protocols) is presented in Section \ref{SectionLowerBound}.

\section{Preliminaries} \label{SectionPreliminaries}

We let \(\N\) denote the set of natural numbers including \(0\) and let \(\Z\) denote the set of integers. We write \(\log n\) for the binary logarithm \(\log_2 n\).

A multiset over a set \(Q\) is a multiplicity function \(f \colon Q \to \N\), which maps every \(q \in Q\) to its number of occurances in the multiset \(f\). We denote multisets using a set notation with multiplicities, i.e.\ \(\Multiset{f(q_1) \cdot q_1, \dots, f(q_m) \cdot q_m}\). We define addition \(f+f'\) on multisets via \((f+f')(q)=f(q)+f'(q)\) for all \(q\in Q\). Multisets are compared via inclusion, defined as \(f \subseteq f' \iff f(q) \leq f'(q)\) for all \(q \in Q\). If \(f \subseteq f'\), then subtraction \(f'-f\) is defined via \((f'-f)(q)=f'(q)-f(q)\) for all \(q \in Q\). The number of elements of \(f\) is denoted \(\Card{f}\) and defined as \(\sum_{f(q) \neq 0} f(q)\) if only finitely many \(q \in Q\) fulfill \(f(q) \neq 0\), and \(\Card{f}=\infty\) otherwise. Elements \(q\) of \(Q\) are identified with the multiset \(\Multiset{1 \cdot q}\). The set of all \emph{finite} multisets over \(Q\) is denoted \(\N^Q\). Given a function \(g \colon A \to B\), its extension \(\hat{g}\) to finite multisets is \(\hat{g} \colon \N^A \to \N^B, \hat{g}(f)=\sum_{f(a) \neq 0} f(a) \cdot \{g(a)\}\).

\begin{definition}
A \emph{protocol scheme} \(\FiniteProt\) is a 5-tuple \((Q,\Alphabet, \Transitions, \Input, \Output)\) of 
\begin{itemize}
\item a (not necessarily finite) set of states \(Q\),
\item a finite input alphabet \(\Alphabet\),
\item a transition function \(\Transitions: Q \times Q \to Q \times Q\),
\item an input mapping \(\Input \colon \Sigma \to Q\),
\item an output mapping \(\Output \colon Q \to \{0,1\}\).
\end{itemize}
\end{definition}

A configuration of \(\FiniteProt\) is a \emph{finite} multiset \(\Config \in \N^Q\). A step \(\Config \to \Config'\) in \(\FiniteProt\) consists of choosing a multiset \(\{q_1, q_2\} \subseteq \Config\) and replacing \(\{q_1, q_2\}\) by \(\{\delta(q_1,q_2)\}\) or \(\{\delta(q_2,q_1)\}\), i.e.\ \(\Config'=(\Config-\{q_1,q_2\}+\{\delta(q_1,q_2)\})\). The intuition is that the configuration describes for every \(q\) the number of agents in \(q\), and a step consists of an agent in \(q_1\) exchanging messages with \(q_2\), upon which these two agents change into the states \(\delta(q_1, q_2)\). Observe that the transition function \(\delta\) distinguishes between the initiator of the exchange and the responder, while in the configuration all agents are anonymous. The number of agents is denoted \(n:=\Card{C}\).

We write \(\to^{\ast}\) for the reflexive and transitive closure of \(\to\), and say that a configuration \(\Config'\) is reachable from \(\Config\) if \(\Config \to^{\ast} \Config'\). A configuration \(\Config\) is initial if there exists a multiset \(w \in \N^{\Sigma}\) such that \(\hat{\Input}(w)=\Config\). In that case \(\Config\) is the initial configuration for input \(w\).  

A configuration \(\Config\) is a \(b\)-consensus for \(b \in \{0,1\}\) if \(\Output(q)=b\) for all \(q\) such that \(\Config(q) \neq 0\), i.e.\ if every state which occurs in the configuration has output b. A configuration \(\Config\) is stable with output \(b\) if every configuration \(\Config'\) reachable from \(\Config\) is a \(b\)-consensus. 

A run \(\rho\) is an infinite sequence of configurations \(\rho=(\Config_0, \Config_1, \dots)\) such that \(\Config_i \to \Config_{i+1}\) for all \(i \in \N\). A run is fair if for all configurations \(\Config\) which occur infinitely often in \(\rho\), i.e.\ such that there are infinitely many \(i\) with \(\Config_i=\Config\), also every configuration \(\Config'\) reachable from \(\Config\) occurs infinitely often in \(\rho\). A run has output \(b\) if some configuration \(\Config_i\) along the run is stable with output \(b\) (and hence all \(\Config_j\) for \(j \geq i\) are also stable with output \(b\)). 

An input \(w \in \N^{\Sigma}\) has output \(b\) if every fair run starting at its corresponding initial configuration \(\hat{\Input}(w)\) has output \(b\). The protocol scheme \(\FiniteProt\) computes a predicate if every input \(w\) has some output. In that case the computed predicate is the mapping \(\N^{\Sigma} \to \{0,1\}\), which maps \(w\) to the output of \(\hat{\Input}(w)\).

\begin{example}
\label{ExampleBinaryRepresentation}
Consider \(Q:=\{0\} \cup \{2^i \mid i \in \N\}\), and define \(\delta(2^i, 2^i)=(2^{i+1}, 0)\), otherwise \(\delta\) is the identity function. Let \(\Sigma=\{x\}\), and let \(x \mapsto 2^0\) be the input mapping. Then a configuration is initial if every agent is in state \(2^0\). Intuitively this protocol will eventually end up with the binary representation of the number of agents. Namely each transition preserves the total sum of all agents' values, and every actual transition (which does not simply leave the agents the same) causes an agent to enter \(0\), so this protocol in fact always reaches a terminal configuration.
For example if we start this protocol with 22 agents we will eventually reach the stable configuration \(\{1\cdot 2^1,1\cdot 2^2,1\cdot 2^4,19\cdot 0\}\), which corresponds to the binary encoding of \(22=10110_2\).
\end{example}

We now define the state complexity of a protocol scheme. A state \(q \in Q\) is \emph{coverable} from some initial configuration \(\Config_0\) if there exists a configuration \(\Config\) reachable from \(\Config_0\) which fulfills \(\Config(q) >0\). The state complexity \(\StateComplexity(n)\) of \(\FiniteProt\) for \(n\) agents is the number of states \(q \in Q\) which are coverable from some initial configuration with \(n\) agents. 

\begin{example}
In the scheme of Example \ref{ExampleBinaryRepresentation}, let \(\Config_n\) be the unique initial configuration with \(n\) agents, i.e.\ \(\Config_n(2^0)=n\) and \(\Config_n(q)=0\) otherwise. For \(n\geq 2\), the states coverable from \(\Config_n\) are exactly \(\{0\} \cup \{2^i \mid i \leq \log n\}\). Hence the state complexity is \(\StateComplexity(n)=\lfloor \log n \rfloor + 2\). 
\end{example}

As defined so far, protocol schemes are not necessarily computable. Hence actual population protocols require some uniformity condition, and that \(S(n)\) is finite for all \(n\). 

\begin{definition}
A \emph{uniform population protocol} \(\FiniteProt=(Q, \Alphabet, \Transitions, \Input, \Output)\) is a protocol scheme s.t. 1) the space complexity \(S(n)\neq \infty\) for all \(n \in \N\) and
2) there is a representation of states as binary strings and linear space Turing-machines (TMs) \(M_{\Transitions}, M_{\Input}, M_{\Output}\), where
\begin{enumerate}
\item \(M_{\Transitions}\): Given (the representation of) two states \(q_1, q_2\), \(M_{\Transitions}\) outputs \(\Transitions(q_1, q_2)\).
\item \(M_{\Input}\): Given multiset \(w\), \(M_{\Input}\) outputs a representation of \(\hat{\Input}(w)\).
\item \(M_{\Output}\): Given a state \(q\) and \(b \in \{0,1\}\), \(M_{\Output}\) checks whether \(\Output(q)=b\).
\end{enumerate}
\end{definition}

We remark that ``linear space'' then in terms of our \(n\), the number of agents, is \(\O(\log \StateComplexity(n))\) space (since the input of the machine is a representation of a state). 

In the literature on uniform population protocols, e.g.\ \cite{ChatzigiannakisOSPS10, ChatzigiannakisMNPS11, DotyEMST18, DotyE19}, often agents are defined as TMs and states hence automatically assumed to be represented as binary strings. We avoid talking about the exact implementation of a protocol via TMs because it introduces an additional logarithm in the number of states and potentially confuses the reader, while most examples are clearly computable.

\begin{example}
In the protocol scheme of Example \ref{ExampleBinaryRepresentation} we represent states by the binary representation of the exponent. Clearly incrementing natural numbers or setting the number to a fixed value are possible by a linear space TM, hence this is a uniform population protocol.
\end{example}

Next we define a more general class of population protocols, which we call weakly uniform. This class includes all known population protocols, and our results also hold for this class, which shows that having a different protocol for every \(n\) does not strengthen the model.

\begin{definition}
A \emph{finite population protocol} is a protocol scheme with a finite set \(Q\).

A \emph{population protocol} \(\Prot\) is an infinite family \((\FiniteProt_n)_{n \in \N}=(Q_n, \Alphabet, \Transitions_n, \Input_n, \Output_n)_n\) of finite population protocols. The state complexity for inputs of size \(n\) is \(\StateComplexity(n):=|Q_n|\).

\(\Prot\) is \emph{weakly uniform} if there exist TMs \(M_{\Transitions}, M_{\Input}, M_{\Output}\) using \(\O(\StateComplexity(n))\) space which:
\begin{enumerate}
\item \(M_{\Transitions}\): Given two states \(q_1, q_2\) and \(n \in \N\) in unary, \(M_{\Transitions}\) outputs \(\Transitions_n(q_1, q_2)\).
\item \(M_{\Input}\): Given multiset \(w\) with \(n\) elements, \(M_{\Input}\) outputs a representation of \(\hat{\Input_n}(w)\).
\item \(M_{\Output}\): Given a state \(q\), \(b \in \{0,1\}\)  and \(n \in \N\) in unary, \(M_{\Output}\) checks whether \(\Output_n(q)=b\).
\end{enumerate}
\label{DefinitionNonUniformProtocol}
\end{definition}

The configurations of \(\Prot\) with \(n\) agents are exactly the configurations of \(\FiniteProt_n\) with \(n\) agents, and accordingly the semantics of steps, runs and acceptance are inherited from \(\FiniteProt_n\). 

The protocol for a given population size \(n\) is allowed to differ completely from the protocol for \(n-1\) agents, as long as TMs are still able to evaluate transitions, input and output. Usually this is not fully utilised, with the most common case of a non-uniform protocol being that \(\log n\) is encoded into the transition function \cite{DotyEGSUS21}. 

Clearly uniform population protocols are weakly uniform. Namely let \(\FiniteProt=(Q, \Alphabet, \Transitions, \Input, \Output)\) be a protocol scheme. Then for every \(n \in \N\) we let \(Q_n\) be the set of states coverable by some initial configuration with \(n\) agents, similar to the definition of state complexity, and define 
\(\FiniteProt_n:=(Q_n, \Alphabet, \delta_n |_{Q_n^2}, \Input, \Output |_{Q_n})\), where \(f |_{A}\) is the restriction of \(f\) to inputs in \(A\). This protocol family computes the same predicate, and is weakly-uniform 
with the same state complexity.

Next we define the complexity classes for our main result. Let \(f \colon \N \to \N\) be a function. \(f\) is space-constructible if there exists a TM \(M\) which computes \(f\) using \(\O(f(n))\) space. 
Given a space-constructible function \(f \colon \N \to \N\), we denote by \(\NSPACE(f(n))\) the class of predicates computable by a 
non-deterministic Turing-machine in \(\O(f(n))\) space. Similarly, let \(\UPP(f(n))\) be the class of predicates computable by uniform population protocols with \(\O(f(n))\) space, and \(\NUPP(f(n))\) be the class of
predicates computable by weakly-uniform population protocols with \(\O(f(n))\) space.

Population protocols decide predicates on multisets \(w\in\N^\Sigma\), or equivalently predicates on \(\N^k\) for \(k=\vert\Sigma\vert\). In order to compare the complexity classes defined on predicates with those defined on languages over an alphabet we define the unary encoding of a predicate \(\varphi\colon\N^k\longrightarrow\binset\) as the language \(L_\varphi\coloneq \left\{1^{x_1}\#1^{x_2}\#\cdots\#1^{x_k} \mid \varphi(x_1,x_2,\ldots,x_k)=1\right\}\). For any complexity class \(\mathcal{C}\) we can now define \(\UENC(\mathcal{C})\coloneq\left\{\varphi\colon\N^k\longrightarrow\binset \mid k\in \N, L_\varphi\in\mathcal{C}\right\}\) as the class of predicates whose unary encoding lies in \(\mathcal{C}\). More specificaly we define \(\UNSPACE(f(n))\coloneq\UENC(\NSPACE(f(n)))\)
\footnote{Previous work has instead used the complexity class \(\SNSPACE(f(n))\) consisting of the symmetric languages (i.e. languages closed under permutation) over the alphabet \(\Sigma\) in \(\NSPACE(f(n))\) to reflect that the agents in a population protocol are unordered. We find it more intuitive to think about a unary encoding with separators, but languages with either encoding can be polynomially reduced to the other.}.
\section{Main Result}\label{sec:mainresult}
We give a characterisation for the expressive power of both uniform and weakly uniform population protocols with $f(n)$ states, where $f\in\Omega(\log n)\cap\O(n^{1-\varepsilon})$, for some $\varepsilon>0$. For technical reasons, we must place two limitations on $f(n)$:
\begin{enumerate}
\item $f(n)=g(\lfloor\log n\rfloor)$ for some $g:\N\rightarrow\N$, i.e.\ $f$ is computable knowing only $\lfloor\log n\rfloor$.
\item $f(n)$ is space-constructible, i.e.\ the function $f$ can be computed in $\SPACE(f(n))$, and 
\item \(f(n)\) is monotonically increasing.
\end{enumerate}
All practically relevant functions fulfill these properties. For the first, we remark that “usually” \(f(n) \in \Theta(f(2^{\lfloor\log n\rfloor}))\).\footnote{The exceptions are plateau functions with large jumps.} For example, while $\sqrt{n}$ is not computable from $\lfloor\log n\rfloor$, we can instead use $\sqrt{2^ {\lfloor\log n\rfloor}}$, which is asymptotically equivalent.

In the remainder of this paper, a function $f$ with these properties is called \emph{reasonable}.

\smallskip
Our bound applies to uniform and weakly uniform protocols. As mentioned in the previous section, the latter includes, to the best of our knowledge, all non-uniform constructions from the literature.

\begin{theorem}\label{TheoremFinal}
Let $\varepsilon>0$ and let $f\in \Omega(\log n) \cap \O(n^{1-\varepsilon})$ be reasonable. Then \[\UPP(f(n))=\NUPP(f(n))=\UNSPACE(f(n) \cdot \log n).\]
\end{theorem}
\begin{proof}
This will follow directly from the upper and lower bounds given by Proposition~\ref{prop:upper} and Theorem~\ref{thm:lower}.
\end{proof}
In particular, we have \(\UPP(\log n)=\NUPP(\log n)=\UNSPACE(\log^2 n)\).
\begin{restatable}{proposition}{propupper}
\label{prop:upper}
Let $\varepsilon>0$ and let $f\in \Omega(\log n) \cap \O(n^{1-\varepsilon})$ be space-constructible. Then 
\[\UPP(f(n)) \subseteq \NUPP(f(n)) \subseteq \UNSPACE(f(n) \log n).\]
\end{restatable}
\begin{proof}
\(\UPP(f(n)) \subseteq \NUPP(f(n))\) follows since uniform protocols are also weakly-uniform.

Hence let \((\FiniteProt_n)_n=(Q_n, \Sigma, \Transitions_n, \Input_n, \Output_n)_n\) be a weakly uniform population protocol computing a predicate \(\varphi\). We have to show that there exists a TM \(M \in \NSPACE(f(n) \log n)\) computing \(\varphi\), when given the input in unary. We employ a similar argument as in the proof of the upper bound in \cite{BlondinEJ19}: First observe that a configuration of \(\FiniteProt_n\) with \(n\) agents can be described by \(|Q_n| \in \O(f(n))\) many numbers up to \(n\), i.e.\ can be stored using \(\O(f(n)\log n)\) bits. Namely one can store the number of agents per state \(q \in Q_n\). The encoding of the initial configuration can easily be calculated by simply counting the ones on the input tape corresponding to each initial state.

Since \(f\) is space-constructible, \(f(n) \log n\) is space-constructible as well. By the Immerman-Szelepcsényi theorem we have \(\NSPACE(f(n) \log n)=\coNSPACE(f(n) \log n)\).

Since the population protocol \((\FiniteProt_n)_n\) computes a predicate, either every fair run starting from the initial configuration \(\hat{\Input}(w)\) accepts or every fair run rejects. \(M\) has to determine which of these is the case. In fact, because every fair run has the same output, we claim that some configuration \(\Config\) reachable from \(\hat{\Input}(w)\) is stable for output \(1\) if and only if \(\hat{\Input}(w)\) is accepted. By definition, an accepting run visits a configuration stable for output \(1\), proving one direction, and for the other direction construct a fair run \(\rho \colon \hat{\Input}(w) \to^{\ast} \Config \to \dots\) by extending \(\hat{\Input} \to^{\ast} \Config\) in a fair way. This run is accepting, and hence also every other fair run is.

We hence construct \(M\) as follows: \(M\) applies \(M_{\Input}\) to obtain a representation of the initial configuration \(\hat{\Input}(w)\). It guesses a configuration \(\Config\), and checks using repeatedly \(M_{\delta}\) that \(\Config\) is reachable from \(\hat{\Input}(w)\). It remains to check that \(\Config\) is stable with output \(1\). A configuration is not stable for output \(1\) if and only if some configuration \(\Config'\) reachable from \(\Config\) contains an agent with output \(0\). Therefore non-stability can be checked in \(\NSPACE(f(n) \log n)\) by guessing \(\Config'\), checking using \(M_{\Output}\) that \(\Config'\) is not a \(1\)-consensus and checking reachability. By Immerman-Szelepcsényi hence also stability is decidable in \(\NSPACE(f(n) \log n)\).
\end{proof}

%
\section{Lower Bound} \label{SectionLowerBound}
\label{sec:lowerbound}

\newcommand{\Qall}{Q'}
\newcommand{\Flags}{F}
\newcommand{\Qinit}{Q_\mathrm{init}}
\newcommand{\Sinit}{S_\mathrm{init}}
\newcommand{\Tinit}{\delta_\mathrm{init}}
\newcommand{\Tcleanup}{\delta_\mathrm{cleanup}}
\newcommand{\Qabs}{Q_\mathrm{abs}}
\newcommand{\Tabs}{\delta_\mathrm{abs}}
\newcommand{\Instruction}{\mathcal{L}}
\newcommand{\InstructionLen}{l}
\newcommand{\Machine}{\mathcal{CM}}

In this section, we prove the following.

\begin{restatable}{theorem}{restateLower}\label{thm:lower}
Let \(\varepsilon>0\) and let \(f\in\Omega(\log n)\cap\O(n^{1-\varepsilon})\) be reasonable. Then \[\UNSPACE(f(n)\log n)\subseteq\UPP(f(n)).\]
\end{restatable}

To do this we first fix a reasonable function \(f\in\Omega(\log n)\cap\O(n^{1-\varepsilon})\) and a predicate \(\varphi \colon \N^\Sigma \longrightarrow \binset\) with \(\varphi \in \UNSPACE (f(n) \log n)\). With a slight modification to the classic 3-counter simulation of a \(f(n)\) space-bounded Turing machine described in \cite{FischerMR68}, we obtain a counter machine \(\Machine\) with \(\vert\Sigma\vert\) input registers and 3 computation registers that decides \(\varphi\) using \(\mathcal{O}\left(2^{f(n)\log n}\right)\) space.

We now construct a population protocol $\Prot=(Q,\Alphabet,\Transitions,\Input,\Output)$ simulating \(\Machine\). There are three main difficulties involved in this construction:

Firstly, in order to sequence multiple operations such that an operation only starts once the previous one has finished, we need a way of performing a \emph{zero-check}, i.e.\ detecting whether an agent with a certain state exists. We achieve this by counting the number of agents using a binary encoding similar as to \cite{BournezCR18}. By keeping track of the agents already seen in an additional counter, we can then perform loops over all agents, and so detect absence of a certain state.

Secondly, we need to encode the counters of \(\Machine\), which can hold values up to \(\mathcal{O}\left(2^{f(n)\log n}\right)\). The two counter encodings described in existing literature of either counting in unary the number of agents in a special state \cite{AngluinAE06}, or the binary encoding of \cite{BournezCR18} both cannot encode numbers this large. We improve on the binary encoding by using \emph{digits} in a higher base of \(\Theta\big(\frac{n}{f(n)}\big)\) and counting in unary within each digit. Manipulating these digits makes heavy use of the looping construct mentioned above.

The final problem, which is inherent to all population protocols, is that an arbitrary number of agents may not participate in any interactions for an arbitrary long amount of time. These errors are detected at some point, but this can happen arbitrarily late. At that point, we solve this by providing a way for the simulation to re-initialize itself.

The protocol consists of multiple phases:
\begin{enumerate}
\item We count the number of agents and initialize the additional counters to zero. This process is detailed in \Cref{sec:init}. \Cref{sec:counter} describes how the counters are manipulated, and \Cref{sec:loops} presents a macro for looping over all agents using the counters for bookkeeping.
\item We set up the \emph{digits}, which encode the counters of \(\Machine\). This phase is described in \Cref{sec:digits}.
\item The instructions of \(\Machine\) are simulated. This is described in \Cref{sec:countermachine}.
\end{enumerate}

As a technical aside, as is usual we assume that the protocol is started with a sufficient number of agents (i.e.\ exceeding some constant). We argue in the proof of Theorem~\ref{thm:lower} why this is not a problem.

\subsection{State Space}
\label{sec:states}
The states will be of the form $Q=\N\times 2^{\Flags}$ for a finite set $\Flags$ of \emph{flags}. A state $(q,S)\in Q$ has \emph{level} $q$. We defer precise definitions until they become relevant.

\smallskip\noindent\textbf{\textsf{Notation.}}
To compactly denote sets of states characterised by flags, we, for example, write \((i,\textsf{Ldr}_0)\) for the set of all level $i$ states which do not include the flag \(\textsf{Ldr}\). In particular, this notation avoids mentioning other flags.

Formally, we write $(i,X^{(1)}_{b_1},...,X^{(k)}_{b_k})$, where $X^{(1)},...,X^{(k)}\in\Flags$ and $b_1,...,b_k\in\{0,1\}$ to refer to the set of all states $(i,S)$ where $S\subseteq\Flags$ fulfils $X^{(j)}\in S\Leftrightarrow b_j=1$ for all $j=1,...,k$. 

On the right-hand side of a transition, we use the same notation with a different meaning: it refers to the state where flags $X^{(1)},...,X^{(k)}$ are as given, and all other flags match the corresponding state that initiated the transition. I.e.\ similar to an assignment command, the mentioned values are set while leaving other flags the same as before.

We also use $*$ as wildcard. On the left-hand side of a transition, it matches anything, and on the right-hand side, it refers to the same value as the corresponding element of the left-hand side. For example, the transition
\[(i,\textsf{Ex}_1),(*,\textsf{Ex}_1)\mapsto(i+1),(*,\textsf{Ex}_0)\qquad\text{for }i\in\N\TraName{example}\]
means that any two agents with flag \textsf{Ex} can interact. The first moves to the next level (with flags unchanged), while the second removes the \textsf{Ex} flag (and leaves its level unchanged).

Sometimes, we want to refer to groups of flags at once, and we write $S_b$ for $S=\{X^{(1)},...,X^{(k)}\}\subseteq\Flags,b\in\{0,1\}$ instead of $X^{(1)}_b,...,X^{(k)}_b$.

\subsection{Initialisation}\label{sec:init}
Our first goal is to reach a configuration with one leader at level \(l_n:=\lfloor \log n \rfloor\), with $l_n+1$ agents each storing one bit of the binary representation of \(n\), and all other agents ``ready to be reset''. Let $b_{l_n}...b_0$ be the binary representation of $n$. Formally we want the leader in state $(l_n,\textsf{Ldr}_1,\textsf{I}_1)$, exactly one counter agent in \((j, \textsf{Ctr}_1, \textsf{N}_{b_j})\) for each \(j \leq l_n= \lceil\log n\rceil\), and all other agents in states $(*,\textsf{Ldr}_0,\textsf{Ctr}_0)$.  The flags $\textsf{Ldr},\textsf{Ctr},\textsf{Free}\in\Flags$ indicate whether the agent is currently a \emph{leader}, a \emph{counter agent}, or \emph{free}, respectively (these are exclusive). Additionally, $\textsf{N},\textsf{I}\in\Flags$, where $\textsf{N}$ indicates whether the bit of the counter is set, and $\textsf{I}$ whether the leader should perform initialisation.

Regarding the input we define $\Input(X):=(0,\{\textsf{Ctr},\textsf{N},X\})$ for $X\in\Alphabet$.

In the counter, the agents perform usual bitwise increments as in Example \ref{ExampleBinaryRepresentation}, though now expressed in terms of the exponent \(i\), and we have to leave one agent in every bit.
\[\begin{aligned}
(i,\textsf{Ctr}_1,\textsf{N}_1),(i,\textsf{Ctr}_1,\textsf{N}_1)&\mapsto (i+1),(i,\textsf{N}_0)&\qquad&\text{for }i\in\N\\
(i,\textsf{Ctr}_1,\textsf{N}_a),(i,\textsf{Ctr}_1,\textsf{N}_b)&\mapsto (i,\textsf{N}_{a+b}),(i,\textsf{Ctr}_0,\textsf{Ldr}_1,\textsf{I}_1)&\qquad&\text{for }i\in\N,a+b\le1
\end{aligned}\TraName{counter}\]
This uses the compact notation for transitions introduced above. Consider the first line. If two agents with value \(i\) are both responsible for the counter and have their \(\textsf{N}\) flag set to \(1\), then, regardless of any other flags, the outcome is as follows: The first agent increments \(i\) (leaving every flag unchanged), and the second agents sets \(\textsf{N}\) to \(0\), again leaving the rest as is.

For the second line, if --- in the same type of encounter --- at most one of the two bits \(\mathsf{N}_a\) and \(\mathsf{N}_b\) was set, then one of the agents unsets his counter flag and becomes a leader with \(\textsf{I}\) flag set to \(1\). 

This is the way for agents to originally set the leader flag. Since we want to have only one leader, we execute a leader election subprotocol. Every time a leader is eliminated, it moves into $\textsf{Free}$, and the remaining leader re-initialises.
\[\begin{aligned}
(i,\textsf{Ldr}_1),(j,\textsf{Ldr}_1)&\mapsto (i,\textsf{I}_1),(0,\textsf{Ldr}_0,\textsf{Free}_1)&&\text{for }i,j\in\N,i\ge j\\
(i,\textsf{Ldr}_1),(j,\textsf{Ctr}_1)&\mapsto (j,\textsf{I}_1),(j)&\qquad&\text{for }i,j\in\N,i<j
\end{aligned}\TraName{leader}\]
The second line causes the leader to eventually point to the most significant bit of \(n\).

\newcommand{\Val}{\operatorname{val}}
Let $\Tinit:=\TraRef{counter}\cup\TraRef{leader}$. For the following proof, as well as later sections, it will be convenient to denote the value of the counter. Given a configuration $C$ and $X\in\Flags$ we write $\Val(C,X):=\sum_{i\in\N}2^iC((i,\textsf{Ctr}_1,X_1))$. For example, the goal of the initialisation is to ensure $\Val(C,\textsf{N})=n$ at all times.

We say that a configuration is \emph{initialised}, if it has
\begin{enumerate}[(1)]
\item exactly one agent in $(l_n,\mathsf{Ldr}_1, \mathsf{Ctr}_0)$,
\item exactly one agent in $(i,\mathsf{Ctr}_1,\mathsf{N}_{b_i})$, for $i=0,...,l_n$ and $b_i$ the $i$-th bit of $n$, and
\item all other agents in $(*,\mathsf{Ldr}_0,\mathsf{Ctr}_0)$.
\end{enumerate}

\begin{restatable}{lemma}{leminit}
\label{lem:init}
Assume that each transition $t \in \delta\setminus\Tinit$ leaves flags $\mathsf{Ldr}, \mathsf{Ctr},\mathsf{N}$ unchanged, and does not affect levels of agents with the \(\mathsf{Ldr}\) or \(\mathsf{Ctr}\) flag. $\Prot$ eventually reaches an initialised configuration with an agent in $(l_n,\mathsf{Ldr}_1,\textsf{I}_1)$, and will remain in an initialised configuration.
\end{restatable}
\begin{proof}
We will show that eventually such a configuration is reached via a \TraRef{leader} transition. Since transitions in $\Tinit$ observe only flags $\textsf{Ldr}, \textsf{Ctr},\textsf{N}$ and levels of leader and counter agents, which by assumption no other transition can change, we can disregard all transitions in $\delta\setminus\Tinit$ for the purposes of this proof.

We have that $\Val(C,\textsf{N})$ is invariant in all reachable configurations $C$, as no transition changes its value. Further, in an initial configuration we have $\Val(C,\textsf{N})=2^0\Abs{C}=n$. Hence the level of any agent with flags \textsf{Ctr} and \textsf{N} is at most $l_n$.

Furthermore, let $n_i$ denote the number of counter agents at level $i$. Then $(n_0,...)$ decreases lexicographically with every \TraRef{counter} transition. As \TraRef{counter} is enabled as long as we have two counter agents on the same level, eventually we will have exactly one agent in \((i, \textsf{Ctr}_1)\) for every \(i=0,...,l_n\), and by the invariant \(\Val(C,\textsf{N})\) the \textsf{N} flag corresponds to the binary representation of \(n\), proving (2).

Therefore eventually no more leaders are created and transition \TraRef{leader} leaves exactly one leader. All other agents are then necessarily in states \((*, \textsf{Ldr}_0, \textsf{Ctr}_0)\), proving (3). Once the last \TraRef{leader} transition occurs, flag \textsf{I} is set on the leader and it has level $l_n$, showing (1).
\end{proof}

\newcommand{\FlagsCounter}{\Flags_{\mathrm{counter}}}
\newcommand{\Tcounter}{\delta_\mathrm{counter}}
\subsection{The Counter}
\label{sec:counter}
We created a counter during initialisation, which now contains the precise number of agents. To perform arithmetic on this counter, we designate a helper agent that executes one operation at a time. This agent uses flags $\FlagsCounter:=\{\textsf{Clr},\textsf{Incr},\textsf{Cmp},\textsf{Swap},\textsf{Done}\}$ to store the operation it is currently executing, and it uses its level to iterate over the bits of the counter. Formally, we say that an agent is a \emph{(counter) helper}, if it has one of the flags in $\FlagsCounter$.

The value stored in the counter using the \textsf{N} flag is immutable (to satisfy the assumptions of Lemma~\ref{lem:init}), so we use flags $\textsf{A},\textsf{B}$ to store two additional values in the counter agents.

The first operation clears the value in \textsf{A}, i.e.\ sets it to zero.
\[\begin{aligned}
(i,\textsf{Clr}_1),(i,\textsf{Ctr}_1)&\mapsto(i+1),(i,\textsf{A}_0)&&\text{for }i\in\N\\
(i+1,\textsf{Clr}_1),(i,\textsf{Ldr}_1)&\mapsto(0,\textsf{Clr}_0,\textsf{Done}_1),(i)\\
\end{aligned}\TraName{clear}\]
It iterates over each bit using the level. To detect that the end has been reached, the helper communicates with the leader, which always has level $l_n$.

To access the value stored in \textsf{B}, we create an operation that swaps it with \textsf{A}. It proceeds in much the same way.
\[\begin{aligned}
(i,\textsf{Swap}_1),(i,\textsf{Ctr}_1,\textsf{A}_a,\textsf{B}_b)&\mapsto(i+1),(i,\textsf{A}_b,\textsf{B}_a)&&\text{for }i\in\N,a,b\in\{0,1\}\\
(i+1,\textsf{Swap}_1),(i,\textsf{Ldr}_1)&\mapsto(0,\textsf{Swap}_0,\textsf{Done}_1),(i)\\
\end{aligned}\TraName{swap}\]
Incrementing is slightly more involved, but only because we do multiple things: we increase the value in \textsf{A} by 1, and then compare it with \textsf{N}. If they match, the value of \textsf{A} is cleared and the helper sets flag $\textsf{R}$ to indicate whether this happened.
\[\begin{aligned}
(i,\textsf{Incr}_1),(i,\textsf{Ctr}_1,\textsf{A}_1)&\mapsto(i+1),(i,\textsf{A}_0)&&\text{for }i\in\N\\
(i,\textsf{Incr}_1),(i,\textsf{Ctr}_1,\textsf{A}_0)&\mapsto(0,\textsf{Incr}_0,\textsf{Cmp}_1),(i,\textsf{A}_1)&&\text{for }i\in\N\\
(i,\textsf{Cmp}_1),(i,\textsf{Ctr}_1,\textsf{A}_a,\textsf{N}_a)&\mapsto(i+1),(i)&&\text{for }i\in\N,a\in\{0,1\}\\
(i,\textsf{Cmp}_1),(i,\textsf{Ctr}_1,\textsf{A}_a,\textsf{N}_{1-a})&\mapsto(0,\textsf{Cmp}_0,\textsf{Done}_1,\textsf{R}_0),(i)&&\text{for }i\in\N,a\in\{0,1\}\\
(i+1,\textsf{Cmp}_1),(i,\textsf{Ldr}_1)&\mapsto(0,\textsf{Cmp}_0,\textsf{Clr}_1,\textsf{R}_1),(i)&&\text{for }i\in\N\\
\end{aligned}\TraName{incr}\]

Let $\Tcounter:=\TraRef{clear}\cup\TraRef{swap}\cup\TraRef{incr}$.
\begin{observation}\label{obs:counter}
Let $\Config$ denote an initialised configuration with exactly one counter helper in state $(0,S)$. If only transitions in $\Tcounter$ are executed, $\Config$ eventually reaches a configuration $\Config'$ with
\begin{enumerate}[(1)]
\item exactly one counter helper in state $(0,S')$, where $S'\cap\FlagsCounter=\{\mathsf{Done}\}$,
\item $\Val(\Config',\mathsf{A})=0$, if $\mathsf{Clr}\in S$,
\item $\Val(\Config',\mathsf{A})=\Val(\Config,\mathsf{B}),$ and $\Val(\Config',\mathsf{B})=\Val(\Config,\mathsf{A}),$ if $\mathsf{Swap}\in S$,
\item $\Val(\Config',\mathsf{A})=\Val(\Config,\mathsf{A})+1$ and $\mathsf{R}\notin S'$, if $\mathsf{Incr}\in S$ and $\Val(\Config,\mathsf{A})+1<\Val(\Config,\mathsf{N})$,
\item $\Val(\Config',\mathsf{A})=0$ and $\mathsf{R}\in S'$, if $\mathsf{Incr}\in S$ and $\Val(\Config,\mathsf{A})+1=\Val(\Config,\mathsf{N})$.
\end{enumerate}
In cases (2), (4), and (5), we also have $\Val(\Config',\mathsf{B})=\Val(\Config,\mathsf{B})$.
\end{observation}
\begin{proof}
Each operation iterates through the bits of the counter and performs the operations according to the above specification. Once the helper reaches level $l_n+1$, we use Lemma~\ref{lem:init} to deduce the existence of a leader at level $l_n$, causing the helper to move to \textsf{Done}. We also remark that the increment operation cannot overflow, as (by specification) $\Val(\Config,\mathsf{A})+1\le\Val(\Config,\mathsf{N})$.
\end{proof}

\subsection{Loops}
\label{sec:loops}
A common pattern is to iterate over all agents. To this end, we implement a loop functionality, which causes a loop body to be executed precisely $n-1$ times.
\[\begin{aligned}
(*,\textsf{Loop}_1,\textsf{Body}_0),(*,\textsf{Done}_1)&\mapsto(*,\textsf{Loop}_0,\textsf{LoopA}_1),(0,\textsf{Done}_0,\textsf{Incr}_1)\\
(*,\textsf{LoopA}_1),(*,\textsf{Done}_1,\textsf{R}_0)&\mapsto(*,\textsf{LoopA}_0,\textsf{Loop}_1,\textsf{Body}_1),(*)\\
(*,\textsf{LoopA}_1),(*,\textsf{Done}_1,\textsf{R}_1)&\mapsto(*,\textsf{LoopA}_0,\textsf{End}_1),(*)\\
\end{aligned}\TraName{loop}\]
This transition is to be understood as a template. Any agent can set flag \textsf{Loop}, and \TraRef{loop} will then interact with the counter, and set flag \textsf{Body}. The agent must then execute another transition removing flag \textsf{Body}, to commence another iteration of the loop. At some point, \TraRef{loop} will instead indicate that the loop is finished, by setting flag \textsf{End}.

\subsection{Cleanup}
\label{sec:cleanup}
After the initialisation of Section~\ref{sec:init}, most agents are in some state in $(*,\textsf{Ldr}_0,\textsf{Ctr}_0)$. We now want to move all of them into state $(0,\{\textsf{Free}\})$, and move the leader to $(l_n,\{\textsf{Ldr},\textsf{Start}\})$. (For intuitive explanations we sometimes elide, as here, the flags corresponding to the input $\Alphabet$, but the transitions take care to not inadvertently clear them.)

During the cleanup, we need one helper agent to perform operations on the counter. The leader will appoint one such agent and mark it using \textsf{Q}. However, it is unavoidable that sometimes such an agent may already exist. Therefore, any counter helper can cause the leader to reset, and during a reset the leader moves any such agents to $(0,\{\textsf{Free},\textsf{T}\})$. Additionally, while resetting the leader sets flag \textsf{T} on any agent it encounters.
\[\begin{aligned}
(*,\textsf{Ldr}_1,\textsf{I}_0),(*,\textsf{Q}_1)&\mapsto(*,\textsf{I}_1),(*)\\
(*,\textsf{Ldr}_1,\textsf{I}_1),(*,\textsf{Ldr}_0,\textsf{Ctr}_0)&\mapsto(*),(0,(\Flags\setminus\Alphabet)_0,\textsf{Free}_1,\textsf{T}_1)\\
(*,\textsf{Ldr}_1,\textsf{I}_1),(*,\textsf{Ctr}_1, \textsf{T}_0)&\mapsto(*),(*,\textsf{T}_1)\\
\end{aligned}\TraName{reset}\]
For the actual cleanup, the leader first appoints one free agent as helper, then uses the loop template from the previous section to iterate over all agents. Free agents are moved to $(0,\{\textsf{Free}\})$, and all other agents are left as-is. At the end of the loop, the helper is moved as well, and the leader enters \textsf{Start}, indicating that cleanup is complete. The following transition \TraRef{cleanup} part 1 is the only transition which unsets the \textsf{I} flag.
\[\begin{aligned}
  (*,\textsf{Ldr}_1,\textsf{I}_1),(*,\textsf{Free}_1)&\mapsto
                                                       \begin{array}{c}
                                                         (*,(\Flags\setminus\Alphabet)_0,\textsf{Ldr}_1,\textsf{Loop}_1),\\
                                                         (0,\textsf{Free}_0,\textsf{Clr}_1,\textsf{T}_1, \textsf{Q}_1)
                                                       \end{array}
  \\
  (*,\textsf{Ldr}_1,\textsf{Body}_1,\textsf{Start}_0),(*,\textsf{T}_1)&\mapsto(*,\textsf{Body}_0),(*,\textsf{T}_0)\\
  (*,\textsf{Ldr}_1,\textsf{End}_1,\textsf{Start}_0),(*,\textsf{Done}_1)&\mapsto(*,\textsf{End}_0,\textsf{Start}_1),(0,(\Flags\setminus\Alphabet)_0,\textsf{Free}_1)\\
\end{aligned}\TraName{cleanup}\]
Now we are ready to prove that eventually the protocol reaches a ``clean'' configuration as in the following lemma. Let $\Tcleanup:=\TraRef{clear}\cup\TraRef{swap}\cup\TraRef{incr}\cup\TraRef{loop}\cup\TraRef{reset} \cup \TraRef{cleanup}$.
\begin{restatable}{lemma}{lemcleanup} \label{lem:cleanup}
Assume that the assumptions of Lemma~\ref{lem:init} hold, and that every transition in $\delta\setminus(\Tinit\cup\Tcleanup)$
\begin{enumerate}[(a)]
\item does not change $\mathsf{I}$ or $\mathsf{Start}$,
\item does not reduce the number of counter helpers,
\item does not use any free agent or counter helper with $\mathsf{T}$ set,
\item does not use any agent with $\mathsf{Ctr}$ set, and
\item does not only use a counter helper or agents in $(*,\mathsf{Free}_1)$ or $(*,\mathsf{Start}_0)$.
\end{enumerate}
Then $\Prot$ eventually reaches an initialised configuration with
\begin{enumerate}[(1)]
\item exactly one agent in $(l_n,\{\mathsf{Ldr},\mathsf{Start}\})$ and $l_n+1$ agents in $(*,\mathsf{Ctr}_1)$, and
\item all other agents in $(0,S\cup\{\mathsf{Free}\})$ for \(S\subseteq\Sigma\), i.e.\ only $\mathsf{Free}$ and input flags are set.
\end{enumerate}
\end{restatable}

\begin{proof}
\newcommand{\Configs}{\mathcal{C}}
Let $\Configs$ denote the set of initialised configurations with in agent in $(*,\textsf{Ldr}_1,\textsf{I}_1)$.

Lemma~\ref{lem:init} guarantees that we reach a configuration $C_1\in\Configs$. As stated there, all configurations reachable from $C_1$ are initialised. We start by arguing that $C_1$ reaches a configuration $C_2$ with exactly one counter helper and one leader with $\textsf{I}$ unset.

First, we note that it is \emph{possible} to reach such a $C_2$, by executing line~2 of \TraRef{reset} to remove all counter helpers, and then executing the first line of \TraRef{cleanup} to create one counter helper and unset $\textsf{I}$. So any fair run from $C_1$ that does not reach such a $C_2$ must avoid configurations in $\Configs$ eventually. (If it visited $\Configs$ infinitely often, by fairness it would have to reach $C_2$ at some point.)

So we now assume that $C_1$ is the last configuration in $\Configs$ on that run. The only possibility to leave $\Configs$ is to have the leader clear $\textsf{I}$, which by assumption~(a) can only be done in \TraRef{cleanup}. This transition creates a counter helper; since we do not reach $C_2$ we thus must have multiple such helpers.

By assumption~(b), the number of counter helpers can only be reduced by a transition in $\Tinit\cup\Tcleanup$. Inspecting these transitions, the only candidates are line 2 of \TraRef{reset} and line 3 of \TraRef{cleanup}. The former is only enabled at configurations in $\Configs$. The latter reduces the number of counter helpers by 1 and sets flag \textsf{Start} on the leader. This flag, by assumption~(a), cannot be cleared by any transition other than line 1 of \TraRef{cleanup}. (Note that line~3 modifies an agent that is not the leader, and there is only one leader since we are operating within initialised configurations.)

Since \textsf{Start} prevents further reductions in the number of counter helpers, at least one such helper remains. Therefore, it is possible to execute the first line of \TraRef{reset} and move back to $\Configs$. By fairness, this happens eventually, contradicting our assumption that $\Configs$ is visited finitely often and $C_2$ not reached, proving our first claim.

Reaching such a $C_2$ must be done by line~1 of \TraRef{cleanup} (since no other transition clears~$\textsf{I}$), which clears the counter and initiates a loop. As we have argued, $C_2$ is initialised and has exactly one counter helper. We now show that all fair runs from $C_2$ either reach $\Configs$ or a configuration $C_3$ fulfilling conditions (1-2).

By assumption~(d), transitions outside of $\Tinit\cup\Tcleanup$ do not interact with the counter, and by~(c) cannot interact with the counter helper (since it has \textsf{T} set). The only transitions involving the counter helper in a state other than \textsf{Done} are $\Tcounter$ and the first line of \TraRef{reset}. Since the latter moves to a configuration in $\Configs$, we may assume wlog that it does not occur.

Inspecting \TraRef{loop}, line 2 of \TraRef{cleanup} is only enabled when the counter helper is in \textsf{Done}. Similarly for line 3 of \TraRef{cleanup}. So when we move the helper to another state, we can apply Observation~\ref{obs:counter} and conclude that it performs its operation correctly. (Transitions outside of $\Tcounter$ may be executed, but cannot affect either the counter or the counter helper.)

This means that line 3 of \TraRef{cleanup} is only executed once line 2 has run exactly $n-1$ times. If $C_2(*,\textsf{Ldr}_0,\textsf{T}_1)<n-1$, this is not possible, and we go back to $\Configs$ eventually using line 1 of \TraRef{reset}. Otherwise, \textsf{T} is set on all non-leader agents and we claim that it was set by lines 2-3 of \TraRef{reset}. Namely note that by assumption~(c) no transition other than \TraRef{cleanup} may use the agents in $(*,\textsf{Free}_1,\textsf{T}_1)$ at all, and by assumption~(e) no transition may be initiated using only the counter helper, the free agents, and the leader without \textsf{Start}.

In that case, all agents with \textsf{T} set must result from lines 2-3 of \TraRef{reset}. Since it resets the (non-input) flags of all non-free agents, the leader will execute line 2 of \TraRef{cleanup} precisely $n-1$ times, and then execute line 3 once, moving to the desired configuration.
\end{proof}

\newcommand{\DigitCount}{g}
\subsection{Digits}
\label{sec:digits}
Let \(\DigitCount\) be the function, such that \(f(x)=\DigitCount\left(\left\lfloor\log x\right\rfloor\right)\) for all \(x\). For the simulation of \(\Machine\), we organise the agents into $f(n)=\DigitCount(l_n)$ many “digits”, which are counters that count up to (roughly) $n/\DigitCount(l_n)$. They do not work by storing the bits individually, as for the counters of the previous section, but instead digit $i$ is stored by having the appropriate number of agents in state $(i,\textsf{Digit}_1,\textsf{N}_1)$.

Overall, the goal is to simulate registers by using multiple digits. For example, consider $k$ digits, where digit $i$ can store a number in $0,...,n_i-1$, and currently stores $d_i$. Then the number stored by this group of digits would be $\sum_{i=1}^k(n_1\cdot...\cdot n_{i-1})d_i$. This is a generalization of standard base \(b\) number systems to allow every digit to have a different base \(n_i\).

In the previous sections, we have made use of a helper agents that could autonomously execute certain tasks (e.g.\ interacting with the counter). We will continue in this vein and designate a new agent for each task.

We start by distributing the free agents into the $\DigitCount(l_n)$ digits. This happens in a simple round-robin fashion.
\[\begin{aligned}
(*,\textsf{Dist}_1),*&\mapsto(0,\textsf{Dist}_0,\textsf{DistA}_1,\textsf{Loop}_1),*\\
(i,\textsf{DistA}_1,\textsf{Body}_1),(0,\textsf{Free}_1,\textsf{V}_0)&\mapsto(i{-}1,\textsf{Body}_0),(i,\textsf{Free}_0,\textsf{Digit}_1,\textsf{V}_1)&&\text{for }i>0\\
(*,\textsf{DistA}_1,\textsf{Body}_1),(*,\textsf{Free}_0,\textsf{V}_0,\textsf{T}_0)&\mapsto(*,\textsf{Body}_0),(*,\textsf{V}_1)\\
(0,\textsf{DistA}_1),(i,\textsf{Ldr}_1)&\mapsto(\DigitCount(i),*),(i)&&\text{for }i\in\N\\
(*,\textsf{DistA}_1,\textsf{End}_1),*&\mapsto(*,\textsf{DistA}_0,\textsf{End}_0, \textsf{DistDone}_1),*\\
\end{aligned}\TraName{dist}\]
We use a new flag \textsf{V} to mark agents that have already been
seen. This ensures that all available agents are distributed. The
restriction to $\textsf{T}_0$ is necessary to satisfy the assumptions
of Lemma~\ref{lem:cleanup} --- but once the cleanup has successfully completed, no agents will have \textsf{T} set.

\newcommand{\Tempset}{M}
Now we implement arithmetic operations on the digits. First, we give a subroutine to detect whether a digit is full (or empty). For the following transition, let $i,j\in\N$, $a,b\in\{0,1\}$ and $M\subseteq\Flags$, with $(j,M)\notin(i,\textsf{Digit}_1,\textsf{M}_a)$ and $(j,M)\in(*,\textsf{U}_{b})$.
\[\begin{aligned}
(*,\textsf{Det}_1),*&\mapsto(*,\textsf{Det}_0,\textsf{DetA}_1,\textsf{Loop}_1,\textsf{R}_0),*\\
\arraycolsep=1.4pt
\begin{array}{r}
(i,\textsf{DetA}_1,\textsf{Body}_1,\textsf{M}_a,\textsf{U}_b),\\(i,\textsf{Digit}_1,\textsf{M}_a,\textsf{U}_{b})
\end{array}
&\mapsto(i,\textsf{Body}_0,\textsf{R}_1),(*,\textsf{U}_{1-b})\\
(i,\textsf{DetA}_1,\textsf{Body}_1,\textsf{M}_a,\textsf{U}_b),(j,M)&\mapsto(i,\textsf{Body}_0),(j,\textsf{U}_{1-b})\\
(i,\textsf{DetA}_1,\textsf{End}_1,\textsf{U}_b),*&\mapsto(i{+}1,\textsf{DetA}_0,\textsf{DetDone}_1,\textsf{End}_0,\textsf{U}_{1-b}),*\\
\end{aligned}\TraName{detect}\]
This is slightly more involved. Similar to before, we mark agents that have been counted (this time using \textsf{U}). To avoid having to do a second loop which resets \textsf{U}, we instead alternate between $\textsf{U}_0$ and $\textsf{U}_1$ every time \TraRef{detect} is executed. In each iteration, we count agents by setting $\textsf{U}$ to the opposite of the value stored in the digit helper. After the loop has completed, the digit helper then flips its own \textsf{U} flag.

To use this routine on digit $i$, we move an agent into $(i,\textsf{Det}_1,\textsf{M}_b)$, where $b$ indicates whether we want to check that the digit is not empty ($b=1$) or not full ($b=0$). The output is returned using the \textsf{R} flag. (For technical reasons, the agent ends in level $i+1$ --- this will be useful when checking multiple digits.)

There are two ways to change the value of a digit $i\in\N$: incrementing and decrementing. Both are analogous, so we only describe the former. The process is straightforward: we check whether digit $i$ is already full; if it is not, we move an agent from $(i,\textsf{Digit}_1,\textsf{M}_0)$ to $(i,\textsf{Digit}_1,\textsf{M}_1)$. Otherwise the digit overflows; we have to set it to $0$ and increment digit $i+1$. (This is simply adding 1 to a number represented using multiple digits in some base.)

Similar to before, let $i,j\in\N$, $b\in\{0,1\}$ and $M\subseteq\Flags$, with $(j,M)\notin(i,\textsf{Digit}_1)$ and $(j,M)\in(*,\textsf{W}_{b})$.
\[\begin{aligned}
  (i,\textsf{DigIncr}_1),(*,\textsf{DetDone}_1)&\mapsto
                                                 \begin{array}{r}
                                                   (i,\textsf{DigIncr}_0,\textsf{DigIncrA}_1),    \\
                                                   (i,\textsf{DetDone}_0,\textsf{Det}_1,\textsf{M}_0)
                                                 \end{array}\\
(*,\textsf{DigIncrA}_1),(*,\textsf{DetDone}_1,\textsf{R}_1)&\mapsto(*,\textsf{DigIncrA}_0,\textsf{DigIncrB}_1),(*)\\
(i,\textsf{DigIncrB}_1),(i,\textsf{Digit}_1,\textsf{M}_0)&\mapsto(0,\textsf{DigIncrB}_0,\textsf{DigDone}_1),(*,\textsf{M}_1)\\
(*,\textsf{DigIncrA}_1),(*,\textsf{DetDone}_1,\textsf{R}_0)&\mapsto(*,\textsf{DigIncrA}_0,\textsf{DigIncrC}_1,\textsf{Loop}_1),(*)\\
(i,\textsf{DigIncrC}_1,\textsf{Body}_1,\textsf{W}_b),(i,\textsf{Digit}_1,\textsf{W}_b)&\mapsto(i,\textsf{Body}_0),(i,\textsf{W}_{1-b},\textsf{M}_0)\\
(i,\textsf{DigIncrC}_1,\textsf{Body}_1,\textsf{W}_b),(j,M)&\mapsto(i,\textsf{Body}_0),(i,\textsf{W}_{1-b})\\
(i,\textsf{DigIncrC}_1,\textsf{End}_1,\textsf{W}_b),*&\mapsto(i+1,\textsf{DigIncrC}_0,\textsf{End}_0,\textsf{DigIncr}_1,\textsf{W}_{1-b}),*\\
\end{aligned}\TraName{digit}\]
We define transitions for \textsf{DigDecr} analogously. 
\subsection{Counter Machine}
\label{sec:countermachine}
\newcommand{\CounterMap}{\nu_i}
\newcommand{\DigPerReg}{K}
In this section, we describe a subprocess that simulates instructions of \(\Machine\), using the digits of the previous section. Each of the $\Abs{\Alphabet}+3$ registers is simulated by $\DigPerReg=\DigitCount(l_n) / (\Abs{\Alphabet}+3)$ digits. We write $\nu_{l_n}(r)$ for the function that maps each register $r$ to its first digit. In particular, $r$ is then simulated by digits $\nu_{l_n}(r),...,\nu_{l_n}(r)+\DigPerReg-1$. Formally, we have $\nu_{l_n}(r):=1+K(r-1)$.
\subsubsection{Input}
Initially, each agent holds one input in $\Alphabet$. We need to initialise the $\Abs{\Alphabet}$ input registers of \(\Machine\) accordingly. We use a loop to make sure that all agents have been moved. However, both the loop and incrementing the digit use the counter stored in \textsf{A} by the \textsf{Ctr} agents; therefore, we swap \textsf{A} and \textsf{B} to switch between them.

Let $X,Y\in\Alphabet$ denote inputs, where $X$ is stored in digits $r,...,r+\DigPerReg-1$.
\[\begin{aligned}
(*,\textsf{Inp}_1,X_1,\textsf{O}_0),(*,\textsf{Done}_1)&\mapsto(*,\textsf{Inp}_0,\textsf{InpA}_1),(0,\textsf{Swap}_1)\\
(*,\textsf{InpA}_1,X_1,\textsf{O}_0),(*,\textsf{DigDone}_1)&\mapsto(*,\textsf{InpA}_0,\textsf{InpB}_1,X_0,\textsf{O}_1),(r,\textsf{DigIncr}_1)\\
(*,\textsf{InpB}_1),(*,\textsf{DigDone}_1)&\mapsto(*,\textsf{InpB}_0,\textsf{InpC}_1),(*)\\
(*,\textsf{InpC}_1),(*,\textsf{Done}_1)&\mapsto(*,\textsf{InpC}_0,\textsf{Inp}_1,\textsf{Loop}_1,\textsf{Body}_0),(*,\textsf{Swap}_1)\\
(*,\textsf{Inp}_1,\textsf{Body}_1,X_1,\textsf{O}_1),(*,Y_1,\textsf{O}_0)&\mapsto(*,X_0,Y_1,\textsf{O}_0),(*,X_1,Y_0,\textsf{O}_1)\\
(*,\textsf{Inp}_1,\textsf{End}_1,\Alphabet_0),*&\mapsto(*,\textsf{InpDone}),*\\
\end{aligned}\TraName{input}\]
There are two considerations complicating the implementation of \TraRef{input}. First, the agent in \textsf{Inp} must count its own input. Second, the overall amount of input flags in the population must not change. We ensure the latter by marking agents with \textsf{O} (instead of e.g.\ consuming the input) and exchanging input flags (second to last line).
\subsubsection{Simulating Instructions}
Finally, we can start simulating the instructions of the counter machine. There are two types of instructions. \textsf{Incr} instructions increment a register and then go nondeterministically to one of two instructions. \textsf{Decr} instructions decrement a register and go to one of two instructions, depending on whether the resulting value is zero. The counter machine accepts by reaching the last instruction.
We make the following assumptions on the behaviour of the counter machine:
\begin{enumerate}[(P1)]
\item\label{ass:1} No increment that would cause an overflow is performed, nor is a decrement on an empty register.
\item\label{ass:2} If it is possible to accept from the initial configuration, every fair run will accept eventually.
\item\label{ass:3} Once reaching the final instruction, the counter machine loops and remains there.
\end{enumerate}
Let $\Instruction_1,...,\Instruction_{\InstructionLen}$ denote the instructions of the counter machine. The subprocess simulating the machine is led by the agent with flag \textsf{CM}; it stores the current instruction using flag $\textsf{IP}^s$, with $s\in\{1,...,\InstructionLen\}$.
Fix some instruction $\Instruction_s=(\mathrm{op},r,s_0,s_1)\in\{\textsf{Incr},\textsf{Decr}\}\times\{1,...,\Abs{\Alphabet}+3\}\times\{1,...,\InstructionLen\}^2$. If $\mathrm{op}=\textsf{Incr}$, we increment counter $\CounterMap(r)$ and move nondeterministically to instruction $s_0$ or $s_1$. Let \(i\in\N,b\in\binset\).
\[\begin{aligned}
(i,\textsf{CM}_1,\textsf{IP}^s_1),(*,\textsf{DigDone}_1)&\mapsto(i,\textsf{IP}^s_0,\textsf{IP}^{s_b}_1),(\CounterMap(r),\textsf{DigDone}_0,\textsf{DigIncr}_1)\\
\end{aligned}\TraName{cm-incr}\]
We remark that the digits have no concept of being grouped into
registers --- if digit $i$ overflows during an increment, the digit
helper moves on to the next digit, even if it “belongs” to a different
register. For our purposes, this is not a problem, since
property~(P\ref{ass:1}) ensures that the last digit of a register
never overflows.

If $\mathrm{op}=\textsf{Decr}$, we decrement counter $\CounterMap(r)$
and check whether it is zero. If so, we move to $s_1$, else
to $s_0$.
\[\begin{aligned}
  (i,\textsf{CM}_1,\textsf{IP}^s_1),(*,\textsf{DigDone}_1)&\mapsto
                                                            \begin{array}{c}
                                                              (i,\textsf{IP}^s_0,\textsf{IPA}^s_1),\\
                                                              (\CounterMap(r),\textsf{DigDone}_0,\textsf{DigDecr}_1)
                                                            \end{array}\\
(i,\textsf{CM}_1,\textsf{IPA}^s_1),(*,\textsf{DigDone}_1)&\mapsto(i,\textsf{IPA}^s_0,\textsf{IPB}^s_1),(*)\\
(i,\textsf{CM}_1,\textsf{IPB}^s_1),(*,\textsf{DetDone}_1)&\mapsto(i,\textsf{IPB}^s_0,\textsf{IPC}^s_1),(\CounterMap(r),\textsf{R}_0)\\
(i,\textsf{CM}_1,\textsf{IPC}^s_1),(j,\textsf{DetDone}_1,\textsf{R}_0)&\mapsto(i),(j,\textsf{DetDone}_0,\textsf{Det}_1,\textsf{M}_1)\\
&\hspace{24.5mm}\text{for }j<\CounterMap(r{+}1)\\
(i,\textsf{CM}_1,\textsf{IPC}^s_1),(\CounterMap(r{+}1),\textsf{DetDone}_1,\textsf{R}_0)&\mapsto(i,\textsf{IPC}^s_0,\textsf{IP}^{s_0}_1),(*)\\
(i,\textsf{CM}_1,\textsf{IPC}^s_1),(*,\textsf{DetDone}_1,\textsf{R}_1)&\mapsto(i,\textsf{IPC}^s_0,\textsf{IP}^{s_1}_1),(*)\\
\end{aligned}\TraName{cm-decr}\]

\subsubsection{Output}
\label{sec:io}
For the population protocol to have an output, we do a standard output broadcast. The agent simulating the counter machine outputs $1$ once the machine has reached the last instruction, and $0$ otherwise. All other agents copy that output.
\[\begin{aligned}
(*,\textsf{CM}_{\InstructionLen}),*&\mapsto(*,\textsf{Output}_1),*\\
(*,\textsf{CM}_1,\textsf{Output}_b),(*)&\mapsto(*),(*,\textsf{Output}_b)&&\text{for }b\in\{0,1\}\\
\end{aligned}\TraName{output}\]
And $\Output((q,S)):=1$ if $\textsf{Output}\in S$, else $\Output((q,S)):=0$.

\subsubsection{Starting the Simulation}
\label{sec:sim-start}
All that remains is initialising the above subprocesses. After cleanup, there will be one unique leader in \textsf{Start} (Lemma~\ref{lem:cleanup}). It creates the subprocesses for the counter and the digits. Then it starts the subprocess that distributes the agents in to the digits. Once that is finished, the leader starts the initialisation of the input registers, and after that, finally starts the counter machine simulation.
\[\begin{aligned}
(*,\textsf{Start}_1,\textsf{Go}_0),(*,\textsf{Free}_1)&\mapsto(*,\textsf{Go}_1,\textsf{GoA}_1),(*,\textsf{Free}_0,\textsf{Done}_1)\\
(*,\textsf{Start}_1,\textsf{GoA}_1),(*,\textsf{Free}_1)&\mapsto(*,\textsf{GoA}_0,\textsf{GoB}_1),(0,\textsf{Free}_0,\textsf{DetDone}_1)\\
(*,\textsf{Start}_1,\textsf{GoB}_1),(*,\textsf{Free}_1)&\mapsto(*,\textsf{GoB}_0,\textsf{GoC}_1),(0,\textsf{Free}_0,\textsf{DigDone}_1)\\
(*,\textsf{Start}_1,\textsf{GoC}_1),(*,\textsf{Free}_1)&\mapsto(*,\textsf{GoC}_0,\textsf{GoD}_1),(0,\textsf{Free}_0,\textsf{Dist}_1)\\
(*,\textsf{Start}_1,\textsf{GoD}_1),(*,\textsf{DistDone}_1)&\mapsto(*),(0,\textsf{DistDone}_0,\textsf{Inp}_1)\\
(*,\textsf{Start}_1,\textsf{GoD}_1),(*,\textsf{InpDone}_1)&\mapsto(*),(0,\textsf{InpDone}_0,\textsf{CM}_1,\textsf{IP}^1)\\
\end{aligned}\TraName{go}\]
This finally allows us to prove Lemma~\ref{thm:lower}:
\restateLower*
\newcommand{\RegCount}{\Gamma}
\begin{proof}
Let $\varphi\in\NSPACE(f(n)\log n)$ denote a predicate, where $\varphi:\N^{\Alphabet}\rightarrow\{0,1\}$. Then there is a $2^{cf(n)\log n}$-bounded counter machine $\Machine$ deciding $\varphi$, for some $c\in\N$, using $\RegCount:=\Alphabet+3$ registers. (The three additional counters are usually used to store the tape left of the head, right of the head, and as a temporary area to perform multiplication and division by constants.)

We may assume that $\Machine$ never exceeds its bounds (ensuring~(P\ref{ass:1})). Further, we can assume that $\Machine$ stores its inputs in some fashion and may nondeterministically restart, as long as it has not accepted. This yields (P\ref{ass:2}). Property~(P\ref{ass:3}) can easily be achieved by a syntactic modification.

Furthermore, it is enough to show that our uniform population protocol \(\Prot\) is correct for all inputs \(\geq n_0\) for some constant \(n_0\) by possibly taking a product with an \(\O(1)\) states population protocol computing \(\varphi\) for small inputs. 

We argue that there is a constant $\beta$, s.t.\ the construction from \Cref{sec:lowerbound} can simulate $\RegCount$ registers that are $2^{cf(n)\log n}$-bounded, using $\DigitCount(l_n):=\beta f(2^{l_n})$ digits in total. Each digit has at least $(n-l_n-5)/\DigitCount(l_n)-1$ agents and there are $\beta f(2^{l_n})/\RegCount$ digits per register. Taking the logarithm, we obtain
\[
\log\Big(\frac{n-l_n-5}{\beta f(2^{l_n})}-1\Big)^{\beta f(2^{l_n})/\RegCount}
\ge \frac{\beta f(n)}{\RegCount}\log\Big(\frac{n}{2\beta\cdot2f(n)}-1\Big)
\ge \frac{\beta f(n)}{\RegCount}\log\Big(\frac{n^\varepsilon}{d\beta}-1\Big)
\]
where $d\in\N$ is a constant s.t.\ $f(n)\le dn^{1-\varepsilon}$. We can further lower-bound this by $\varepsilon\beta/\RegCount\cdot f(n)\log n-\O(1)$. Choosing a suitably large constant $\beta$, this is at least $cf(n)\log n$, as desired.

It remains to argue that our construction is correct. Using lemmas~\ref{lem:init} and~\ref{lem:cleanup}, we know that the protocol eventually reaches a configuration with exactly one leader, a counter initialised to $n$, and all other agents in a well-defined state. Afterwards, at each step at most one agent can execute a transition, and correctness follows from careful inspection of the transitions defined above.
\end{proof}


\section{Conclusion}

We have characterised the expressive power of population protocols with $f\in\Omega(\log n)\cap\O(n^{1-\varepsilon})$ states. This closes the gap left open by prior research for uniform protocols, and gives the complexity for protocols with $\Theta(\log n)$ or $\Theta(\polylog n)$ states --- the most common constructions in the literature. Our characterisation applies to both uniform and non-uniform protocols.

The upper bound uses the Immerman-Szelepcsényi theorem to argue that a nondeterministic space-bounded Turing machine can simulate the protocol and determine whether it has stabilised. Similar arguments can be found in the literature \cite{BlondinEJ19}.

Our construction is more involved. It uses the standard idea of determining the total number of agents and then performing zero-checks, i.e.\ checking whether a state is absent by iterating over all agents. Using zero-checks, it is straightforward to simulate counter-machines. There are two main difficulties: First, with only $\O(\log n)$ states, no single agent can store $n$. Instead, we have to distribute that information over multiple agents (namely those with flag \textsf{Ctr}), and those agents must collaborate to perform computations on that number. Second, it is neither sufficient to use a constant number of counters with $n$ agents, nor to use $f(n)$ counters with constant number of agents (i.e.\ bits). We must do both at the same time, which results in the \textsf{Digit} agents. This is one main point where our construction improves upon \cite{BournezCR18} and prevents the loss of log factors.

We have focused on the expressive power of protocols that can run for an arbitrary amount of time. However, time-complexity plays an important role, and many constructions in the literature focus on being fast. Does limitting the running time affect the expressive power? We conjecture that such protocols can be modelled well by randomised, space-bounded Turing machines, but it is unclear whether one can obtain a characterisation in that case.

One important result about constant-state population protocols is the decidability of the verification problem~\cite{EsparzaGLM17} --- a natural question is whether this result can be extended to, e.g.\ protocols with $\Theta(\log n)$ states. Unfortunately, our characterisation answers this question in the negative. This does open the question of whether there exist subclasses that exclude our construction (and may, therefore, have a decidable verification problem), but include known constructions from the literature for e.g.\ the majority predicate.

Finally, one gap remains for \emph{non-uniform} (or weakly uniform) protocols with $o(\log n)$ states. In particular, is it possible to decide a non-semilinear predicate with $o(\log n)$ states? We conjecture $\UNL \coloneq \UENC(\NL)\subseteq\NUPP(\log\log n)$, i.e.\ there is a (non-uniform) population protocol with $\O(\log\log n)$ states for every predicate in $\UNL$, in particular for \(x \cdot y=z\) or for deciding whether a given input \(x\) is a prime number.

\bibliography{expressive_power_of_PP.bib}

\end{document}